\documentclass[11pt]{article}
\usepackage[margin=1in]{geometry}
\usepackage[colorlinks=true, allcolors=blue]{hyperref}

\usepackage[usenames,dvipsnames]{xcolor}
\usepackage{amsmath, booktabs}
\usepackage{complexity}
\usepackage{amsthm}
\usepackage{tikz}
\usetikzlibrary{shapes}
\usetikzlibrary{positioning}
\usetikzlibrary{quotes, calc}
\usetikzlibrary{arrows.meta}
\usepackage{subcaption}

\newcommand{\symDegree}[0]{\Delta^{*}}
\tikzset{
  filled vertex/.style = {circle,draw=blue,fill=black!50,inner sep=1.5pt},
  empty vertex/.style = {circle, draw, fill = white, inner sep=1.5pt, minimum width=1.5pt},
  corner/.style  = {fill=blue,inner sep=2pt},
  selected solution/.style = {fill=RubineRed, draw=blue, inner sep=2pt},
  unselected solution/.style = {draw=blue, diamond, fill=cyan, inner sep=1.5pt}
}

\theoremstyle{plain} \newtheorem{theorem}{Theorem}
\newtheorem{lemma}[theorem]{Lemma}
\newtheorem{proposition}[theorem]{Proposition}
\newtheorem{corollary}[theorem]{Corollary}

\title{Feedback Set Problems on Bounded-Degree (Planar) Graphs}
\author{
  Tian Bai\thanks{School of Computer Science and Engineering, University of Electronic Science and Technology of China, Chengdu, China. {\tt myxiao@uestc.edu.cn, tianbai@std.uestc.edu.cn}.}
  \and
  Yixin Cao\thanks{Department of Computing, Hong Kong Polytechnic University, Hong Kong, China.  {\tt yixin.cao@polyu.edu.hk}.
  This work was done while Yixin Cao was visiting Tian Bai and Mingyu Xiao at the University of Electronic Science and Technology of China.
  } 
  \and Mingyu Xiao\footnotemark [1]
}

\begin{document}

\maketitle

\begin{abstract}
  The feedback set problems are about removing the minimum number of vertices or edges from a graph to break all its cycles. Much effort has gone into understanding their complexity on planar graphs as well as on graphs of bounded degree. We obtain a complete complexity classification for these problems on bounded-degree digraphs, including the planar case.
  In particular, we show that both problems are  $\NP$-complete on digraphs of maximum degree three, while on planar digraphs the feedback vertex set problem is polynomial-time solvable when each vertex has either indegree at most one or outdegree at most one, and  $\NP$-complete otherwise.
  We also give tight degree bounds for the connected feedback vertex set problem on undirected graphs, both planar and non-planar. We close the paper with a historical account of results for feedback vertex set on undirected graphs of bounded degree.
    
\end{abstract}

\section{Introduction}

In the feedback vertex set problem, we are given a graph $G$ and an integer $k$, and are asked whether there exists a set of at most $k$ vertices whose removal from $G$ makes the graph acyclic.
The input graph may be undirected or directed. There is also an edge/arc version, in which one deletes edges in undirected graphs or arcs in digraphs.
We refer to these collectively as feedback set problems, and their solutions are \emph{feedback vertex sets}, \emph{feedback edge sets}, and \emph{feedback arc sets}.
For undirected graphs, the edge version is easy, because one only needs to keep a spanning forest.
All the remaining variants are $\NP$-complete, and both directed variations appear in Karp's list~\cite{karp-72-reduction}.

The classic reduction for the feedback vertex set problem on digraphs, which was attributed to the Algorithms Seminar at Cornell University, was from the vertex cover problem~\cite{karp-72-reduction}.
In the vertex cover problem, we are given a graph $G$ and an integer $k$, and asked whether there exists a set of at most $k$ vertices whose removal from $G$ makes an edgeless graph.
Since parallel edges and directions do not impact the vertex cover problem, we may always assume that $G$ is a simple undirected graph.
In this reduction, we
\begin{quote}
  replace each edge of $G$ with two arcs of opposite directions.
\end{quote}
Since these two arcs already form a cycle, any feedback vertex set of the resulting digraph must contain one of the two ends, hence a vertex cover of $G$.
The second reduction, devised by Lawler and Karp, reduces the feedback vertex set problem to the feedback arc set problem on digraphs:
\begin{quote}
  replace each vertex $v$ with two new vertices $v^-$ and $v^+$, connected by an arc $v^- v^+$, and reassign all the incoming arcs toward $v$ to $v^-$ and all the outgoing arcs from $v$ to $v^+$.
\end{quote}
In the new digraph, each new vertex has either indegree at most one or outdegree at most one.
In terms of breaking directed cycles, removing the arc $v^- v^+$ has the same effect as removing either $v^-$ or $v^+$ in the new digraph, which is the same as removing $v$ from the original digraph.
Throughout this paper, we will refer to these two constructions as the \emph{doubling operation} and the \emph{splitting operation}.
We can easily adapt the doubling operation to the feedback vertex set problem on undirected graphs.
Instead of two arcs, we use two parallel edges between the same pair of adjacent vertices.\footnote{If simple graphs are demanded, we may subdivide one of the pair of edges; or equivalently, for each edge $u v$, we introduce a new vertex $x$ and make it connected to $u$ and $v$ only.}

Since feedback set problems have wide applications, a natural question is their complexity on special graph classes.
Much attention has been given to planar graphs.  A graph is \emph{planar} if it can be drawn in the plane so that edges intersect only at common endpoints.
Indeed, in the very first systematic study of the feedback set problems, Younger~\cite{younger-63-thesis} conjectured that the feedback arc set problem can be solved in polynomial time on planar graphs.
This was later confirmed by Lucchesi and Younger~\cite{lucchesi-76-thesis, lucchesi-78-minimax}; see also \cite{lovasz-76-minimax}.
In contrast, the feedback vertex set problem remains  $\NP$-hard on planar graphs because the vertex cover problem is  $\NP$-hard on planar graphs~\cite{garey-77-rectilinear-steiner-tree}, and the doubling operation preserves planarity. (The splitting operation does not generally preserve planarity; see more discussions below.)

Another natural question on the feedback set problems is to identify the threshold between tractability and intractability, with respect to the maximum degree of the input graph.
It has been observed that the vertex degrees are intrinsic to this problem~\cite{cao-18-ufvs}.
We use $d(v)$ to denote the degree of a vertex $v$, i.e., the total number of edges/arcs incident to $v$, and the maximum degree of the graph is 
\[
  \Delta(G) = \max_{v\in V(G)} d(v).
\]
On the one hand, all variants are polynomial-time solvable when the maximum degree is at most two.
On the other hand, since the vertex cover problem remains $\NP$-complete on planar graphs of maximum degree three~\cite{garey-77-rectilinear-steiner-tree}, the reduction by the doubling operation implies that the feedback vertex set problem remains $\NP$-complete on planar graphs, directed or not, of maximum degree six (see Section~\ref{sec:digraphs}).
These simple observations left a small gap, for which the literature contains incomplete and occasionally inaccurate claims.

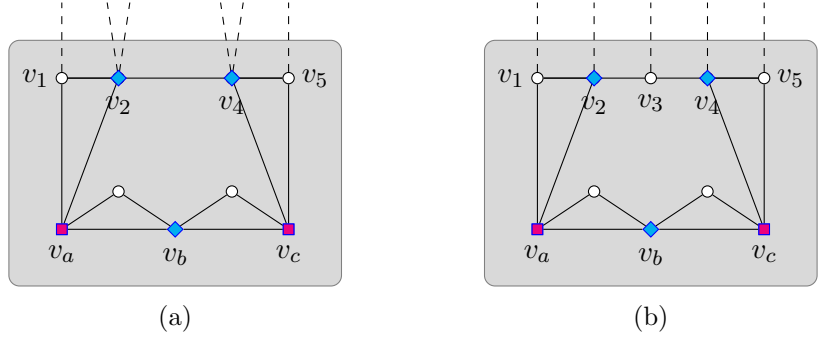
\begin{figure}[h]
  \centering
  \begin{subfigure}[b]{.35\linewidth}
    \centering
    \begin{tikzpicture}
      \draw[fill = gray!30, draw = gray, rounded corners] (-2.2, 3.5) rectangle (2.2, .25);
      \node [unselected solution, "$v_b$" below] (o) at (0, 1) {};
      \foreach [count=\i] \x in {-1, 1} {
        \foreach \j in {1, 2} {
          \node [empty vertex] (t\x\j) at (1.5*\x/\j, 3) {};
        }

        \node [selected solution] (b\x) at (1.5*\x, 1) {};
        \node [empty vertex] (e\x) at (.75*\x, 1.5) {};
        \draw (b\x) -- (t\x 1) -- (t\x 2) -- (b\x);
        \draw (o) -- (b\x) -- (e\x) -- (o);
        \foreach \j in {-1, 1} {
          \draw[dashed] (t\x 2) -- ++(.15*\j, 1);
        }
        \draw[dashed] (t\x 1) -- ++(0, 1);        
}
      \draw (t-11) -- (t11);
      \foreach \x/\i in {t-12/2, t12/4} {
        \node[unselected solution, "$v_{\i}$" below] at (\x) {};
      }
      \foreach \x/\i in {b-1/a, b1/c} {
        \node[selected solution, "$v_{\i}$" below] at (\x) {};
      }
      \foreach \x/\i in {-1/1, 1/5} {
        \node[] at (1.85*\x, 3) {$v_{\i}$};
      }
    \end{tikzpicture}
    \caption{}
  \end{subfigure}  
  \quad
  \begin{subfigure}[b]{.35\linewidth}
    \centering
    \begin{tikzpicture}
      \draw[fill = gray!30, draw = gray, rounded corners] (-2.2, 3.5) rectangle (2.2, .25);
      \node [unselected solution, "$v_b$" below] (o) at (0, 1) {};
      \foreach [count=\i] \x in {-1, 1} {
        \foreach \j in {1, 2} {
          \node [empty vertex] (t\x\j) at (1.5*\x/\j, 3) {};
          \draw[dashed] (t\x\j) -- ++(0, 1);
        }

        \node [selected solution] (b\x) at (1.5*\x, 1) {};
        \node [empty vertex] (e\x) at (.75*\x, 1.5) {};
        \draw (b\x) -- (t\x 1) -- (t\x 2) -- (b\x);
        \draw (o) -- (b\x) -- (e\x) -- (o);
}
      \draw (t-11) -- (t11);
      \foreach \x/\i in {t-12/2, t12/4} {
        \node[unselected solution, "$v_{\i}$" below] at (\x) {};
      }
      \foreach \x/\i in {b-1/a, b1/c} {
        \node[selected solution, "$v_{\i}$" below] at (\x) {};
      }
      \foreach \x/\i in {-1/1, 1/5} {
        \node[] at (1.85*\x, 3) {$v_{\i}$};
      }
      \node [empty vertex, "$v_3$" below] (v3) at (0, 3) {};
      \draw[dashed] (v3) -- ++(0, 1);
    \end{tikzpicture}
    \caption{}
  \end{subfigure}  
  \caption{Speckenmeyer's gadgets (a) for a degree-six vertex and (b) for a degree-five vertex \cite{speckenmeyer-83-thesis-fvs}.
    The edges incident to the original vertex are assigned to the dashed lines in cyclic order.
    Note that a minimum feedback vertex set contains at least two vertices from each gadget (which have to be $v_a$ and $v_c$ if only two), and at most three vertices (no cycle visits any remaining vertex after vertices $v_2$, $v_4$, and $v_b$ removed).
  }
  \label{fig:speckenmeyer}
\end{figure}

For undirected graphs, Speckenmeyer and two other groups~\cite{ueno-88-cubic-fvs, furst88} closed the gap in the 1980s.
In his Ph.D. thesis~\cite{speckenmeyer-83-thesis-fvs}, Speckenmeyer devised two simple gadgets, reproduced in Figure~\ref{fig:speckenmeyer} (with a slight modification), to replace vertices of degree six and five.
As a result, the feedback vertex set problem remains  $\NP$-complete for undirected planar graphs of maximum degree at most four.
It is a nice exercise to reduce the feedback vertex set problem on undirected graphs of maximum degree three to the same problem on \emph{undirected cubic graphs}, i.e., undirected graphs in which every vertex has degree three.
Speckenmeyer~\cite{speckenmeyer-83-bound-cubic-fvs,speckenmeyer-88-fvs-and-nsis} observed that the feedback vertex set problem and the connected vertex cover problem, where the solution is required to induce a connected subgraph, are equivalent on undirected cubic graphs.
This completes the classification for undirected graphs since the latter problem can be solved in polynomial time on undirected cubic graphs \cite{ueno-88-cubic-fvs, furst88}.\footnote{See
the appendix
for a more detailed historical account of these results.}

The situation for digraphs is murkier.
It is natural to attempt to adapt Speckenmeyer's gadgets to planar digraphs.
However, while we do not need to distinguish edges incident to a vertex in an undirected graph, the arcs in a digraph have directions.
As we will see, the application of Speckenmeyer's approach crucially relies on the \emph{embedding}, i.e., how the graph is drawn on the plane.
It seems unlikely that one can design a planar replacement gadget for a vertex whose incoming and outgoing arcs alternate around the embedding.
This obstruction is reminiscent of the non-planarity of $K_{3, 3}$.
Recall that the reduction by the doubling operation starts from the vertex cover problem on planar cubic graphs.
For each pair of opposite arcs between the same pair of vertices, we can switch their order independently from other pairs.
In other words, we can assign any orientation to the cycle of length two formed by them.
If we can arrange the embedding such that the three cycles of each vertex are in different orientations, i.e., one or two clockwise and the other(s) anticlockwise, we can handle them.
In particular, for each vertex in such an embedding, two incoming arcs are consecutive, and two outgoing arcs are consecutive.
We can deal with them in the same spirit as the grouping of vertices used in Figure~\ref{fig:speckenmeyer}(a).

For planar digraphs of maximum degree three, the feedback vertex set problem and the feedback arc set problem turn out to be equivalent.
We can delete vertices with only incoming arcs or only outgoing arcs because they do not lie on any directed cycles.
Every remaining vertex then has either indegree one or outdegree one, which makes the problem equivalent to feedback arc set.
Since planar feedback arc set is polynomial-time solvable, so is planar feedback vertex set in this case.

The polynomial-time algorithm for planar digraphs of maximum degree three does not apply to non-planar digraphs of maximum degree three, which is $\NP$-complete by a simple reduction. (This is different from undirected graphs.)
After the splitting operation is applied to a digraph, we further split vertices $v^-$ (resp., $v^+$) into a path and let each vertex on the path receive one incoming (resp., outgoing) arc to $v$.
The same reduction applies to the feedback arc set problem.
This reduction is surprisingly simple, and more surprising is that it has not been documented, to the best of our knowledge.
It is stronger than claims in the literature, e.g., \cite{garey-79}.

Table~\ref{tbl:main-results} summarizes the resulting complexity classification. 
For planar digraphs, the relevant parameter is slightly more refined than the maximum degree.
For a vertex $v$ of a digraph $D$, let~$d^-(v)$ and $d^+(v)$ denote  its indegree and outdegree, respectively.
Let~$d(v) = d^-(v) + d^+(v)$ and
\[
  \symDegree(D) = \max_{v\in V(D)} \min\{ d^-(v), d^+(v) \}.
\]
Thus, $\symDegree(D) \le 1$ means that every vertex has either indegree at most one or outdegree at most one.
In this setting, the minimum feedback vertex sets and minimum feedback arc sets have the same cardinality.
This is precisely the idea behind the splitting operation.  While $\symDegree(D)$ is not useful for general digraphs, it captures the tractability threshold in the planar case.
Note that $\Delta(D) \le 3$ implies $\symDegree(D) \le 1$, while $\symDegree(D) \ge 2$ implies $\Delta(D) \ge 4$.
All the hardness results remain even if the graphs are required to be bipartite, as we can always subdivide every edge or arc, which has no impact on the minimum solutions.

\begin{table}[h]
  \centering
  \caption{The complexity landscape of the feedback set problems on graphs of bounded degree.
  }
  \label{tbl:main-results}
  \begin{tabular}{r| l r | l r}
    \toprule
    & \multicolumn{2}{c|}{$\P$} & \multicolumn{2}{c}{$\NPC$}
    \\
    \midrule
    Undirected, vertex & $\Delta \le 3$ & \cite{speckenmeyer-88-fvs-and-nsis, ueno-88-cubic-fvs, furst88} & $\Delta \ge 4$ & \cite{speckenmeyer-83-thesis-fvs} 
    \\
    Planar undirected, vertex & $\Delta \le 3$ & \cite{speckenmeyer-88-fvs-and-nsis, ueno-88-cubic-fvs, furst88} & $\Delta \ge 4$ & \cite{speckenmeyer-83-thesis-fvs} 
    \\
    \midrule
    Directed, vertex & $\Delta \le 2$ & Proposition~\ref{lem:digraph-maximum-two} & $\Delta \ge 3$ & Theorem~\ref{thm:degree-bounded} 
    \\
    Directed, arc & $\Delta \le 2$ & Proposition~\ref{lem:digraph-maximum-two} & $\Delta \ge 3$ & Theorem~\ref{thm:degree-bounded}
    \\
    Planar directed, vertex  & $\symDegree \le 1$ & Corollary~\ref{cor:planar-digraph} & $\symDegree \ge 2$ & Theorem~\ref{PDFVS：NPC}
    \\
    Planar directed, arc & always & \cite{lucchesi-76-thesis, lucchesi-78-minimax} & &
    \\    \bottomrule 
  \end{tabular}
\end{table}

In passing let us remark that
Garey and Johnson~\cite{garey-79} already mentioned bounds on degrees for the feedback set problems on planar digraphs.
However, they provided neither proofs nor references for these claims.
We do not know what arguments they had in mind.
In particular, our proof of Theorem~\ref{thm:degree-bounded} does not appear to follow directly from their statements, since they listed different degree bounds for the vertex and arc versions.
They only claimed that the feedback arc set problem is $\NP$-complete for digraphs with indegree and outdegree both at most three.

Additionally, we settle the complexity of the connected variation of the feedback vertex set problem on undirected graphs.
In the connected feedback vertex set problem, the solution is required to induce a connected subgraph.
While it is trivial on an undirected graph of maximum degree two, it is already $\NP$-hard on planar undirected graphs of maximum degree eight, which follows from the same reduction above (by doubling edges) and the fact that the connected vertex cover problem remains $\NP$-hard even for undirected planar graphs with maximum degree four~\cite{garey-77-rectilinear-steiner-tree}.
The main difficulty in narrowing the gap is to design a low-degree gadget that preserves both planarity and the connectivity requirement on the solution.
Revisiting the gadget of Bodlaender et al.~\cite{comgeoBodlaenderFGPSW09} on the weighted version of this problem, we show that the problem is already $\NP$-complete on planar undirected graphs of maximum degree three.

The vertex set and edge/arc set of a graph $G$ are denoted by, respectively, $V(G)$ and $E(G)$.
Since membership in NP is immediate for all problems considered, we focus on hardness proofs.

\section{Digraphs of bounded degree}
\label{sec:digraphs}

We begin with the easy case of maximum degree two.
We may remove all vertices that have no incoming arcs or no outgoing arcs: they are not visited by any directed cycle.
The remaining graph consists of disjoint cycles, and it suffices to take one vertex or arc from each cycle.

\begin{proposition}\label{lem:digraph-maximum-two}
  Both feedback set problems can be solved in polynomial time on digraphs of maximum degree at most two.
\end{proposition}

On the other hand, the following na\"{i}ve reduction suffices to pin down the hardness when there are vertices of degree three.
The other end of an incoming (resp., outgoing) arc to $v$ is an \emph{in-neighbor} (resp., \emph{out-neighbor}) of $v$.
We have seen that splitting a vertex $v$ into $v^-$ and $v^+$, with an arc from $v^-$ to $v^+$, separates in-neighbors and out-neighbors of $v$.
If~$v$ had more than two in-neighbors, we would further split $v^-$ into multiple vertices so that each receives one incoming arc.
We process the vertex $v^+$ in a similar way.
We show that each original vertex~$v$ is represented by a path gadget~$P(v)$, and that choosing~$v$ in the original instance corresponds to cutting the unique ``bridge'' between the incoming and outgoing parts of~$P(v)$.  This preserves both the existence and the size of minimum solutions.
This reduction works for both the vertex and arc variations and leads to the same bound on total degrees. 

\begin{figure}[h]
  \centering
  \begin{subfigure}[b]{.4\linewidth}
    \centering
    \begin{tikzpicture}
      [
      > = {Stealth[width = 3.7pt, length = 6pt]},
      ]
      \node[filled vertex, "$v$"] (v) at(0, 1.25) {};
      \foreach \y/\il/\ir in {0/1/1, 1/2/2, 2/3/3, 5/d^-(v)/d^+(v)} {
        \node[empty vertex, "$x_{\il}$" left] (x\y) at(-1.25, {(5-\y)/2}) {};
        \node[empty vertex, "$y_{\ir}$" right] (y\y) at(1.25, {(5-\y)/2}) {};
        \draw[->] (x\y) edge (v) (v) edge (y\y);
      }
      \foreach \i in {-1, 1}
      \node at(1.5*\i, .9) {$\vdots$};
    \end{tikzpicture}
    \caption{}
  \end{subfigure}  
  \quad
  \begin{subfigure}[b]{.4\linewidth}
    \centering
    \begin{tikzpicture}
      [
      > = {Stealth[width = 3.7pt, length = 6pt]},
      ]
      \draw[fill = gray!30, draw = gray, rounded corners] (-1.75, 3) rectangle (1.75, -.5);
\foreach \y/\il/\ir in {0/1/1, 1/2/2, 2/3/3, 5/d^-(v)/d^+(v)} {
        \node[empty vertex, "$v^-_{\il}$" right] (x-\y) at(-1.25, {(5-\y)/2}) {};
        \node[empty vertex, "$v^+_{\ir}$" left] (y-\y) at(1.25, {(5-\y)/2}) {};
\draw[<-, dashed] (x-\y) -- ++(-1, 0);
        \draw[->, dashed] (y-\y) -- ++(1, 0);
      }
      \foreach \l in {x, y} {
        \draw[->] (\l-0) edge (\l-1) (\l-1) edge (\l-2) (\l-2) edge ++ (0, -.4);
        \draw[<-] (\l-5) edge ++ (0, .4);
      }
      \foreach \i in {-1, 1} {
        \node at(1.25*\i, .9) {$\vdots$};
        \node at(2.25*\i, .9) {$\vdots$};
      }
      \foreach \x in {x-5, y-0}
      \node[filled vertex] at (\x) {};
      \draw[->, thick] (x-5) -- (y-0);
    \end{tikzpicture}
    \caption{}
  \end{subfigure}  
  \caption{A vertex $v$ in (a) is replaced by a directed path on $d(v)$ new vertices in (b).
}
  \label{fig:degree-reduction}
\end{figure}
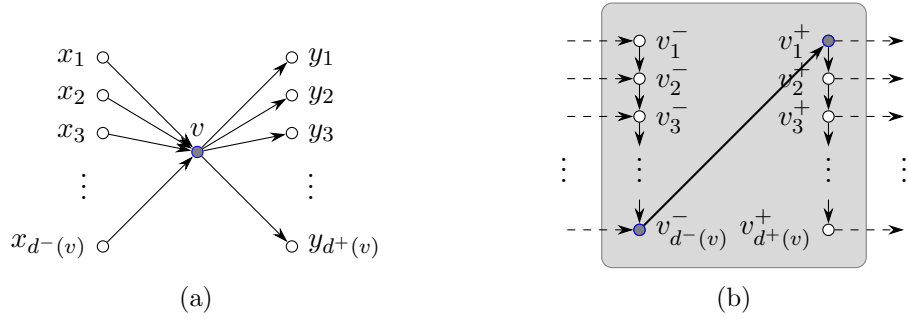

\begin{theorem}\label{thm:degree-bounded}
  Both feedback set problems remain $\NP$-complete for digraphs of maximum degree at most three.
\end{theorem}

\begin{proof}
  Let $D$ be a digraph.
We transform $D$ into another digraph $D'$ by performing  the following operations on each vertex $v$ of $D$.
  We introduce a directed path, denoted as $P(v)$, on $d(v)$ new vertices:
  \[
    v^{-}_{1} v^{-}_{2} \cdots v^{-}_{d^{-}(v)}
    v^{+}_{1} v^{+}_{2} \cdots v^{+}_{d^{+}(v)}.
  \]

  Let the in-neighbors and out-neighbors of $v$ be numbered in an arbitrary way as $\{x_{1}, x_{2}, \ldots, x_{d^{-}(v)}\}$ and
  $\{y_{1}, y_{2}, \ldots, y_{d^{+}(v)}\}$, respectively.
  It is possible that $x_i = y_j$ for some $i$ and $j$ with $1\le i\le d^{-}(v)$ and $1\le j\le d^{+}(v)$.
  For $i = 1, 2, \ldots, d^{-}(v)$, we add an arc $x_{i} v^{-}_{i}$, and for $i = 1, 2, \ldots, d^{+}(v)$, we add an arc $v^{+}_{i} y_{i}$.  
  See Figure~\ref{fig:degree-reduction} for an illustration.
  Let $D'$ denote the resulting digraph.

  If we remove the vertex $v^{+}_{1}$ from the gadget $P(v)$, then no cycle can visit the remaining vertices of this gadget.  Thus, a minimum feedback vertex set of $D'$ contains at most one vertex from a gadget.  Moreover, if a minimum feedback vertex set of $D'$ contains a vertex from a gadget $P(v)$, we can always replace it with $v^{+}_{1}$.
  We claim that a set $S\subseteq V(D)$ is a feedback vertex set of $D$ if and only if 
  \[
    S' = \{v^{+}_1\mid v\in S\}
  \]
  is a feedback vertex set of $D'$.

  For necessity, suppose for contradiction that $C'$ is a cycle in $D' - S'$.
  Since each gadget is acyclic, the cycle~$C'$ visits vertices from more than one gadget.
  Moreover, if $C'$ visits vertices on $P(v)$, then $C'$ must visit the arc $v^{-}_{d^{-}(v)} v^{+}_{1}$.
Consider the cyclic sequence of gadgets visited by $C'$.
  By the construction of $D'$, if $P(u)$ is followed by $P(v)$, then $u v\in E(D)$.
  We can thus obtain a cycle of $D$.
  Since $S$ is a feedback vertex set of $D$, it contains at least one vertex on this cycle, say $v$.
  But then $v^+_1\in S'$, a contradiction.
  
For sufficiency, suppose for contradiction that $C$ is a cycle in $D - S$.
  Let $u v$ be an arc on $C$.
  Suppose that $u$ is the $p$th in-neighbor of $v$
  and $v$ is the $q$th out-neighbor of $u$
  (in the numbering of $N^+(u)$ and $N^-(v)$ during the reduction).
  Hence, $u^{+}_{q} v^{-}_{p}$ is an arc of $D'$.
  We replace $u v$ with the path
  \[
    u^{-}_{d^{-}(u)} u^{+}_{1} u^{+}_{2} \cdots u^{+}_{q} v^{-}_{p} v^{-}_{p+1} \cdots v^{-}_{d^{-}(v)}.
  \]
  Conducting similar replacements for all the arcs on $C$ leads to a cycle $C'$ of $D'$.
  Since $S'$ is a feedback vertex set of $D'$, it contains at least one vertex on $C'$, say $v^+_1$.
  But then $v\in S$, a contradiction.

  The same reduction also works for the feedback arc set problem:
a vertex set $S\subseteq V(D)$ is a feedback vertex set of $D$ if and only if the arc set $\{v^{-}_{d^{-}(v)} v^{+}_{1}\mid v\in S\}$ is a feedback arc set of $D'$.
\end{proof}

Note that the reduction above does not start from the feedback set problems on digraphs of bounded degree.
Theorem~\ref{thm:degree-bounded} is slightly stronger than separate bounds on maximum in- and out-degrees.

\begin{corollary}Both feedback set problems remain $\NP$-complete for digraphs with indegree and outdegree both at most two.
\end{corollary}
\begin{proof}
  Vertices with indegree zero or outdegree zero lie on no directed cycle and may therefore be deleted.  In the remaining digraph, every degree-3 vertex has indegree/outdegree pattern $(2,1)$ or $(1,2)$, so both indegree and outdegree are at most two.
\end{proof}

\section{Planar digraphs of bounded degree} \label{PDFVS}

A sharp-eyed reader may have noticed that the gadget in Figure~\ref{fig:degree-reduction} does preserve planarity.
However, this happens only when the arcs incident to the vertex $v$ are embedded in a certain way.
A plane embedding of a digraph~$D$ is \emph{bipolar} if, around every vertex, the incoming arcs appear consecutively in the circular order.  Equivalently, the outgoing arcs are also consecutive.
Note that when applied to such an embedded digraph, the splitting operation ends with a planar digraph.
 
\begin{theorem}\label{thm:bipolar-embeddings}
  The feedback vertex set problem can be solved in polynomial time in planar digraphs that admit a bipolar embedding.
\end{theorem}
\begin{proof}
  Let $D$ be a planar digraph that admits a bipolar embedding.
Let $D'$ denote the resulting digraph obtained from~$D$ by the splitting operation.
  Because incoming and outgoing arcs are consecutive around each vertex, the split vertices~$v^{-}$ and~$v^{+}$ of~$v$ can be inserted locally without creating crossings.
Then a subset $S\subseteq V(D)$ is a feedback vertex set of $D$ if and only if $\{v^{-}v^{+} \mid v\in S\}$ is a feedback arc set of $D'$~\cite{karp-72-reduction}.
  Since $D'$ is planar, the feedback arc set problem can be solved in polynomial time. 
  The statement follows.
\end{proof}

It may be worth mentioning that the algorithm does not require a bipolar embedding as part of the input.
It suffices to apply the splitting operation and then find a plane embedding of the resulting digraph.
A planar digraph $D$ with $\symDegree(D) \le 1$ trivially admits a bipolar embedding: for every vertex $v$, either $d^{-}(v) \le 1$ or $d^{+}(v) \le 1$. 

\begin{corollary}\label{cor:planar-digraph}
  The feedback vertex set problem can be solved in polynomial time on planar digraphs~$D$ with~$\symDegree(D) \le 1$.
\end{corollary}

Since~$\Delta(D) \le 3$ implies $\symDegree(D) \le 1$, the next statement follows immediately from Corollary~\ref{cor:planar-digraph}.

\begin{corollary}
  The feedback vertex set problem can be solved in polynomial time on planar digraphs~$D$ with $\Delta(D) \le 3$.
\end{corollary}

In the rest, we show that both corollaries are tight.
We say that a digraph $D$ is \emph{$3$-regular} if $d^-(v) = d^+(v) = 3$ for every vertex $v\in V(D)$.
Since the {vertex cover} problem remains $\NP$-complete on planar cubic graphs~\cite{garey-77-rectilinear-steiner-tree}, the reduction by the doubling operation implies that the feedback vertex set problem remains $\NP$-complete on $3$-regular planar digraphs.
A natural attempt is to devise planar gadget(s) \`a la Speckenmeyer to replace the vertices one by one.
However, we have to respect the cyclic ordering of the arcs incident to each embedded vertex, which is not always easy. 
In particular, it seems very unlikely if the incoming and outgoing arcs incident to a vertex alternate in the embedding.

If we use ``$-$'' and ``$+$'' to denote an incoming arc and an outgoing arc, respectively, then the cyclic ordering of the six arcs incident to any vertex can be written as $(-, -, -, +, +, +)$ in a bipolar embedding.  Likewise, we write $(-, +, -, +, -, +)$ if they alternate.
Up to reversal and cyclic symmetry, there are four possible circular orders of three incoming and three outgoing arcs:
\[
  (-, -, +, +, -, +)  \text{ and } (+, +, -, -, +, -).
\]
Note that they are symmetric.
We say that a plane embedding of a $3$-regular planar digraph~$D$ is \emph{irregular} if for each vertex $v \in V(D)$, the arcs incident to $v$ are in the circular order of either $(-, -, +, +, -, +)$ or $(+, +, -, -, +, -)$.  
Interestingly, we can always construct an irregular embedding for a $3$-regular planar digraph obtained by the doubling operation.
Note that this is not true for general $3$-regular planar digraphs; a classic result states that if the underlying graph is three-connected, the digraph admits an essentially unique embedding.
The trick here is that we can freely switch the order of any pair of opposite arcs between the same pair of vertices.

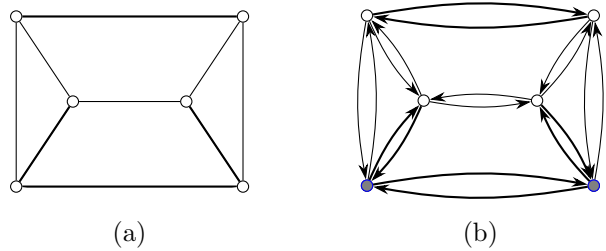
\begin{figure}[h]
  \centering
  \begin{subfigure}[b]{.25\linewidth}
    \centering
    \begin{tikzpicture}[scale = 1.5]
      \foreach \x/\l in {-1/x, 1/y} {
        \foreach \i in {0, 1} 
        \node[empty vertex] (\l\i) at (1*\x, 1.5*\i) {};
        \node[empty vertex] (\l2) at (.5*\x, .75) {};
        \draw (\l 1) -- (\l 2) -- (\l 0) -- (\l 1);
        \draw[thick] (\l 0) -- (\l 2);
      }
      \foreach \i in {0, 1} {
        \draw[thick] (x\i) -- (y\i);
      }
      \draw (x2) -- (y2);

\draw[ultra thin, white, bend right=10] (x0) edge (y0);      
    \end{tikzpicture}
    \caption{}
  \end{subfigure}  
  \quad
  \begin{subfigure}[b]{.25\linewidth}
    \centering
    \begin{tikzpicture}
      [
      scale = 1.5,
> = {Stealth[width = 3.7pt, length = 6pt]},
      ]
      \foreach \x/\l in {-1/x, 1/y} {
        \foreach \i in {0, 1} 
        \node[empty vertex] (\l\i) at (1*\x, 1.5*\i) {};
        \node[empty vertex] (\l2) at (.5*\x, .75) {};
        \foreach \i in {0} {
\draw[->, bend right = 10] (\l 0) edge (\l 1) (\l 1) edge (\l 0);
          \draw[->, bend right = 10] (\l 1) edge (\l 2) (\l 2) edge (\l 1);
          \draw[thick, ->, bend left=10] (\l 0) edge (\l 2) (\l 2) edge (\l 0);
        }
      }
      \foreach \l in {x, y} \node[filled vertex] at (\l 0) {};
      \foreach \i in {0, 1} {
        \draw[thick, ->, bend left = 10] (x\i) edge (y\i) (y\i) edge (x\i);
      }
      \draw[->, bend right = 10] (x2) edge (y2) (y2) edge (x2);
    \end{tikzpicture}
    \caption{}
  \end{subfigure}  
  \caption{(a) An undirected planar cubic graph $G$, where the edges in $F$ are represented as thick lines, and (b) an irregular embedding of the digraph obtained from $G$ by the doubling operation.
Note that the two shadowed vertices have $(-, -, +, +, -, +)$ embedding, while all the others $(+, +, -, -, +, -)$.
  }
  \label{fig:irregular-embedding}
\end{figure}

\begin{lemma}\label{lem:irregular-embedding}
  Let $G$ be an undirected planar cubic graph.
  The digraph $D$ obtained by the doubling operation admits an irregular embedding, and such an embedding can be constructed in polynomial time.
\end{lemma}
\begin{proof}
  We fix an arbitrary plane embedding of $G$.
  According to Duncan et al.~\cite{compgeomDuncanEK04}, $E(G)$ can be covered by two linear forests (the union of disjoint paths), and they can be found in linear time.
  Let~$F$ be one of the linear forests.
  We do the doubling operation as follows.
  For each edge~$e\in E(G)$, we replace~$e$ by a pair of opposite arcs. We choose their local order so that the resulting 2-cycle is clockwise exactly when~$e\in F$ (see Figure~\ref{fig:irregular-embedding}).
  Let $D$ denote the resulting digraph, which is clearly a planar digraph: the construction actually ends with an embedding.
  To see that the embedding is irregular, note that each vertex $v$ is incident to one or two edges in $F$.
  If all three edges incident to $v$ belong to $F$, then $F$ is not a linear forest; if all of them are disjoint from $F$, then another linear forest cannot cover the three edges.
  If $v$ is incident to exactly one edge from $F$, then the embedding is $(+, +, -, -, +, -)$; it is $(-, -, +, +, -, +)$ otherwise.
  This concludes the proof.
\end{proof}

\begin{figure}[h]
  \centering
  \begin{subfigure}[b]{.27\linewidth}
    \centering
    \begin{tikzpicture}
      [
      > = {Stealth[width = 3.7pt, length = 6pt]},
      ]
      \def\radius{1.5}
      \def\n{20}
      \node[filled vertex, "$v$"] (v) at (0, 0) {};
      \foreach[count=\i] \a in {12, 13, 0} {
        \pgfmathsetmacro{\angle}{- \a * (360 / \n)}        
        \node[empty vertex] (x\i) at (\angle:\radius) {};
        \draw[->] (x\i) -- (v);
        \node at ({\angle}:{\radius+.3}) {$x_\i$};
      }
      \foreach[count=\i] \a in {17, 18, 10} {
        \pgfmathsetmacro{\angle}{- \a * (360 / \n)}        
        \node[empty vertex] (y\i) at (\angle:\radius) {};
        \draw[->] (v) -- (y\i);
        \node at ({\angle}:{\radius+.3}) {$y_\i$};
      }
      \node at (0, -1) {};
    \end{tikzpicture}
    \caption{}
  \end{subfigure}  
  \quad    
  \begin{subfigure}[b]{.45\linewidth}
    \centering
    \begin{tikzpicture}
      [
      > = {Stealth[width = 3.7pt, length = 6pt]},
      ]
      \draw[fill = gray!30, draw = gray, rounded corners] (-3, 2.2) rectangle (3, -1.5);

      \node[selected solution, fill=RubineRed, "$v_c$" below] (c) at (0, 0) {};
      \foreach [count=\i] \x/\l/\s in {-1/x/-, 1/y/+} {
        \node[empty vertex, "$v^{\s}_{a}$"] (\l-1) at (\x*1.5, 1.5) {};
        \node[empty vertex, "$v^{\s}_{12}$"] (\l-2) at (\x*2.5, 1.5) {};
        \node[empty vertex] (\l-3) at (\x*1.5, 0) {};
        \node[empty vertex] (\l-4) at (\x*.75, 0.75) {};
        \node[empty vertex] (\l-5) at (\x*-2.5, -.75) {};
      }
      \foreach \from/\to in {x-2/x-1, x-3/x-1, x-3/c, c/x-4, x-4/x-3, x-1/y-5, y-5/x-3,
        x-1/y-1, x-5/y-5, 
        y-1/y-2, y-3/x-5, x-5/y-1, y-1/y-3, c/y-3, y-4/c, y-3/y-4} {
        \draw[->] (\from) -- (\to);
      }
      \foreach \x/\s in {y/+, x/-} {
        \node[empty vertex, "$v^{\s}_{3}$" below] at (\x-5) {};
      }      
      \foreach \x in {-1, 1} {
        \draw[->, dashed] (y-2) -- ++(1, .2*\x);
        \draw[<-, dashed] (x-2) -- ++(-1, .2*\x);
      }
      \draw[->, dashed] (y-5) -- ++(-1, 0);
      \draw[<-, dashed] (x-5) -- ++(1, 0);
      
      \foreach \x in {x-1, x-5} {
        \node[selected solution, fill=RubineRed] at (\x) {};
      }
      \foreach \x/\l in {x/b, y/d} {
        \node[unselected solution, fill=cyan, "$v_{\l}$" below] at (\x-3) {};
        \node["$v'_{\l}$"] at (\x-4) {};
      }
    \end{tikzpicture}
    \caption{}
  \end{subfigure}  
  \caption{(a) A vertex $v$ with $(-, -, +, +, -, +)$ embedding, and (b) the gadget for $v$.}
  \label{fig:gadget-planar-dfvs}
\end{figure}
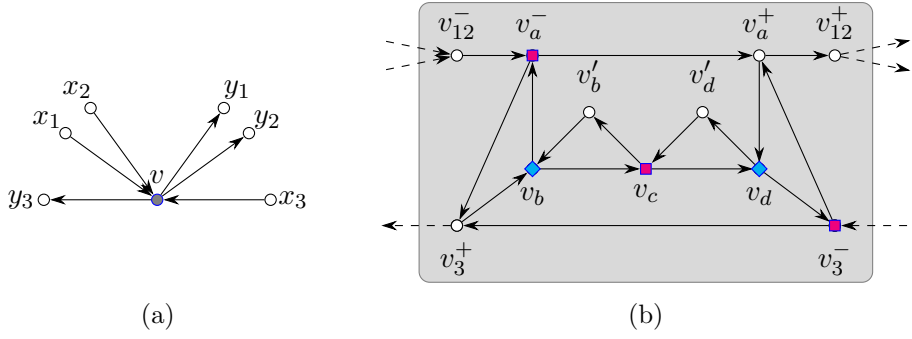

An irregular embedding allows us to group the pair of consecutive incoming arcs and the pair of outgoing arcs for the gadget.  In Figure~\ref{fig:gadget-planar-dfvs},  the gadget has two ``entry'' vertices $v^-_{12}, v^-_{3}$ and two ``exit'' vertices $v^+_{12}, v^+_{3}$.  It is designed so that every feedback vertex set contains either the pair $\{v_{b}, v_{d}\}$ or at least three vertices corresponding to selecting~$v$.

\begin{theorem} \label{PDFVS：NPC}
  The feedback vertex set problem remains $\NP$-complete for planar digraphs with indegree and outdegree both at most two.
\end{theorem}
\begin{proof}
  Let $D$ be the digraph derived from an undirected planar cubic graph by the doubling operation, and let~$k$ be an integer.
  Note that it is $\NP$-complete to decide whether $(D, k)$ is a yes-instance of the feedback vertex set problem.
  We use Lemma~\ref{lem:irregular-embedding} to construct an irregular embedding of $D$.
  For each vertex $v\in V(D)$ with embedding $(-, -, +, +, -, +)$, we replace $v$ by the gadget $H(v)$, which consists of 11 vertices
  \[
    v^-_{12}, v^+_{12}, v^-_{3}, v^+_{3}, v^-_{a}, v^+_{a}, v_{b}, v'_{b}, v_{c}, v_{d}, v'_{d},
  \]
  and the following $16$ arcs:
  \begin{itemize}
\item an arc $v^-_{3} v^+_{3}$, a directed path $v^-_{12} v^-_{a} v^+_{a} v^+_{12}$; and
\item four directed triangles $v_b v_c v'_b$, $v_c v_d v'_d$, $v^-_a v^{+}_3 v_{b}$, and $v^{-}_3 v^{+}_a v_{d}$.
  \end{itemize}
Finally, we reassign the arcs incident to $v$ as
\[
    x_1 v^-_{12}, x_2 v^-_{12}, v^+_{12} y_1, v^+_{12} y_2, x_3 v^-_{3}, \text{ and }v^+_{3} y_3.
  \]
  See Figure~\ref{fig:gadget-planar-dfvs} for an illustration.
  We replace each vertex $v$ in $V(D)$ with embedding $(+, +, -, -, +, -)$ in a symmetric manner.
  Let $D'$ denote the resulting digraph.
  Note that $D'$ is a planar digraph.

  Let $n = |V(D)|$ and $k' = 2n + k$. 
  We argue that $(D, k)$ is a yes-instance of the feedback vertex set problem if and only if $(D', k')$ is.

  For necessity, suppose that $S$ is a feedback vertex set of $D$.
  We claim that
  \[
    S' = \{ v_b, v_d \mid v \in V(D) \setminus S \} \cup \{ v^-_{a}, v_c, v^-_{3} \mid v \in S \}
  \]
  is a feedback vertex set of $D'$.
  Note that 
  \[
    |S'| = 2|V(D) \setminus S| + 3|S| = 2|V(D)| + |S| = 2n + k.
  \]
  After removing either $\{v_{b}, v_{d}\}$ or $\{v^-_{a}, v_c, v^-_{3}\}$ from the gadget $H(v)$, no directed cycle can visit the remaining vertices from $H(v)$.  
  Thus, $D' - S'$ is acyclic.

  For sufficiency, suppose that $S'$ is a feedback vertex set of $D'$.
  We claim that
  \[
    S = \{v\mid S' \text{ contains more than two vertices from } H(v)\}
  \]
  is a feedback vertex set of $D$.
  Since the triangles $v^-_a v_{b} v^{+}_3$ and $v^{-}_3 v^{+}_a v_{d}$ are disjoint, $S'$ contains at least two vertices from each gadget.
  Thus, each gadget corresponding to a vertex in~$S$ contributes at least three selected vertices, whereas every other gadget contributes at least two, 
  \[
    |S'| \ge 3 |S| + 2 |V(D)\setminus S| = 2 n + |S|,
  \]
  and hence $|S|\le k$.

Suppose for contradiction that $C$ is a cycle in $D - S$.
  Let $u v w$ be a sub-path of $C$ (possibly $u = w$).
  Since $v\not\in S$, the feedback vertex set $S'$ contains $v_b$ and $v_d$ and no other vertex in $H(v)$.
  Then there are paths from $v^-_{i}$ to $v^+_{j}$ in $H(v)$ for all pairs of $i, j\in \{12, 3\}$ via one of the following
  \begin{itemize}
  \item $v^-_{12} v^-_{a} v^+_{a} v^+_{12}$;
  \item $v^-_{12} v^-_{a} v^+_{3}$;
  \item $v^-_{3} v^+_{a} v^+_{12}$; or
  \item $v^-_{3} v^+_{3}$.
  \end{itemize}
  Thus, each arc of~$C$ can be simulated by an internal directed path through the corresponding gadget, and we can concatenate them to form a cycle of $D' - S'$, a contradiction.
  The proof is now complete.
\end{proof}

\section{The connected feedback vertex set problem} \label{PCFVS}

The connected feedback vertex set problem has been studied primarily for undirected graphs, and we also restrict attention to that setting.
Again, it is trivial to solve this problem in polynomial time when the maximum degree is two.
On the other hand, since the connected vertex cover problem remains $\NP$-hard even for undirected planar graphs with maximum degree four~\cite{garey-77-rectilinear-steiner-tree}, the same reduction (by doubling edges) implies the $\NP$-hardness of the connected feedback vertex set problem on undirected planar graphs of maximum degree eight.
Bodlaender et al.~\cite{comgeoBodlaenderFGPSW09} showed that the weighted version of the connected feedback vertex set problem is $\NP$-hard on planar graphs of maximum degree four (even only two different weights allowed).

\begin{figure}[h]
  \centering
  \begin{subfigure}[b]{.28\linewidth}
    \centering
    \begin{tikzpicture}
        \def\radius{1}

        \foreach[count=\x from 0] \l in {u, v} {
          \pgfmathsetmacro{\shift}{2*\x}        
          \begin{scope}[xshift=\shift cm]
            \node[empty vertex] (\l) at(0, 0) {$\l$};
        \foreach \i in {1, ..., 4} {
          \pgfmathsetmacro{\angle}{-\i * (360 / 4)}
          \draw (\l) -- node[midway, fill=white, text=gray, inner sep=1pt] {\small $\i$} (\angle:\radius);
        }
        \end{scope}
      }
      \node at (0, -1.44) {};
\end{tikzpicture}
    \caption{}
  \end{subfigure}  
  \begin{subfigure}[b]{.65\linewidth}
    \centering
    \begin{tikzpicture}
      \draw[fill = cyan!30, draw = gray, rounded corners] (1.2, -1.) rectangle (5.8, 1);
      \def\radius{.8}

      \foreach[count=\x from 0] \l in {u, v} {
        \pgfmathsetmacro{\shift}{7*\x}        
        \begin{scope}[xshift=\shift cm]
        \def\n{8}
        \node[fill = gray!30, circle, draw = gray, minimum width = 20mm] {};
        \foreach[count=\i] \a in {1, ..., 8} {
          \pgfmathsetmacro{\angle}{(-.5 - \a) * (360 / \n)}
          \node[filled vertex] (\l\i) at (\angle:\radius) {};
\node at ({\angle}:{\radius-.27}) {\small $\l_\i$};
        }
        \foreach[count=\i] \a in {1, ..., 7} {
          \pgfmathsetmacro{\j}{int(\i+1)}        
          \draw (\l\i) -- (\l\j);
        }
      \end{scope}
    }

    \foreach \x/\d/\pos in {u8/+1/right, u7/+1/right, v3/-1/left, v4/-1/left} {
      \draw (\x) -- ++ (1.85*\d, 0) node (\x 0) {};
      \foreach[count=\i] \displace in {0.8, 1.7} {
\coordinate[\pos = \displace cm of \x] (\x\i) {};
      }
    }
    \foreach \x/\l in {u71/1, u72/2, v41/{8 n}, v42/{8 n - 1}} {
      \node[filled vertex, "$e_{uv}^{\l}$"] at (\x) {};
    }
    \foreach \x/\l in {u81/{8 n}, u82/{8 n - 1}, v31/1, v32/2} {
      \node[empty vertex, "$e_{vu}^{\l}$" below] at (\x) {};
    }
    
    \draw[dashed] (u80) -- (v30) (u70) -- (v40);

    \foreach \x/\d in {u5/1, v5/1, u2/-1, v2/-1} {
      \draw[dashed] (\x) -- ++ (0, .82*\d);      
    }
    \foreach \x/\d in {u6/1, v6/1, u1/-1, v1/-1} {
      \draw[dashed] (\x) -- ++ (0, .82*\d);      
    }
    \foreach \x/\d in {v7/1, v8/1, u3/-1, u4/-1} {
      \draw[dashed] (\x) -- ++ (.82*\d, 0);
    }
    \end{tikzpicture}
    \caption{}
  \end{subfigure}
  \caption{(a) Two adjacent vertices of $G$ and edges incident to them, where the numbers on edges denote their cyclic order.
  For example, the edge $u v$ is the fourth of $u$ and the second of $v$.
  (b) The gadgets $P(u)$, $P(u v)$, and $P(v)$, which comprise $8$, $16 n$, and $8$ vertices, respectively.}
  \label{FIG: PCFVS}
\end{figure}
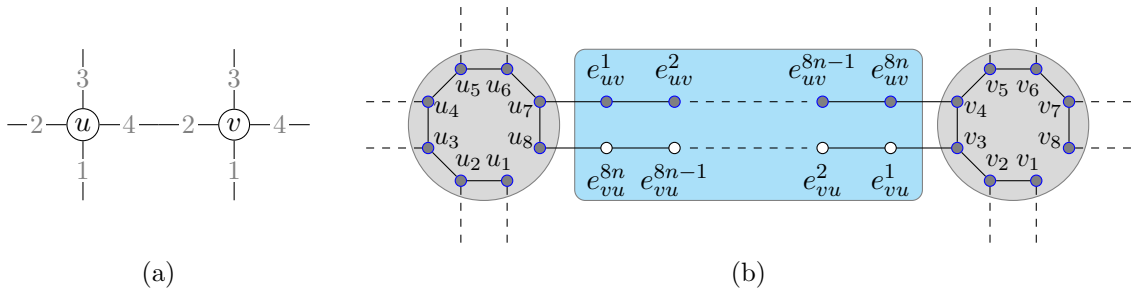

\begin{theorem}
  The connected feedback vertex set problem remains $\NP$-complete on undirected planar graphs of maximum degree three.
\end{theorem}
\begin{proof}
  Let $G$ be a planar graph of maximum degree at most four, and let $n = |V(G)|$.
We may assume without loss of generality that $G$ is connected, and it contains at least two edges.
For each vertex $v\in V(G)$, we introduce a path $P(v)$ on eight new vertices $v_{1}v_{2}\cdots v_{8}$, drawn around the position of $v$ in clockwise order.
  For each edge $u v\in E(G)$, we introduce two disjoint paths each on $8 n$ new vertices:
    \[
      e^1_{u v} e^2_{u v} \cdots e^{8 n}_{u v}
      ~\text{and}~
      e^1_{v u} e^2_{v u} \cdots e^{8 n}_{v u},
    \]
    denoted as~$P(uv)$ and~$P(vu)$, respectively.
    We fix an arbitrary embedding of $G$ and connect these gadgets as follows.
    Suppose the neighbors of $v$ are $x_1$, $x_2$, $\ldots$, $x_{d(v)}$ in the clockwise order, where~$d(v) \le 4$.
    For $i = 1, 2, \ldots, d(v)$, we add edges
    \[
      v_{2 i - 1} e^1_{v x_i}
      ~\text{and}~
      v_{2 i} e^{8 n}_{x_i v}.
    \]
    See Figure~\ref{FIG: PCFVS} for an illustration.

    Let $G'$ denote the resulting graph.
    It is clear from construction that $G'$ is a connected planar graph of maximum degree three. 
    We claim that $(G, k)$ is a yes-instance of the connected vertex cover problem if and only if $(G', k')$ is a yes-instance of the connected feedback vertex set problem, where
    \[
       k' = 8k + 8 n (k - 1).
    \]

    For necessity, suppose $S$ is a connected vertex cover of $G$.  We may assume that $|S| = k$: we can always add vertices to it because $n > k$.
    Since $G[S]$ is connected, we can find a spanning tree~$T$ of $G[S]$; we root it at an arbitrary vertex.
We claim that
    \[
        S' = \{ v_{1}, v_{2}, \ldots, v_{8} \mid v \in S\} \cup \{e^{1}_{u v}, e^{2}_{u v}, \ldots, e^{8 n}_{u v} \mid u \text{ is the parent of } v \}
    \]
    is a connected feedback vertex set of $G'$.
    Note that
    \[
      |S'| = 8 |S| + 8 n (|{S}| - 1) = 8k + 8 n (k - 1)= k',
    \]
    and the subgraph of $G'$ induced by $S'$ is connected, because the chosen edge gadgets correspond to a spanning tree of $G[S]$, and each such gadget attaches to the full paths $P(v)$ for its endpoints.

    Since~$S$ is a vertex cover of $G$, for every original edge~$u v$, at least one of $u,v$ belongs to~$S$.  Hence, after deleting~$S'$, each remaining component is of one of the following two forms:
    \begin{itemize}
    \item a single path~$P(v u)$ for some edge~$u v$; or
    \item the vertex gadget $P(v)$ together with some attached sub-paths~$P(v u)$.
    \end{itemize}
In both cases the component is a tree (in fact, a path or a subdivided star), and therefore acyclic.
    Thus, $S'$ is a connected feedback vertex set of $G'$.

    For sufficiency, we show that
    \[
        S = \{ v \mid S'~\text{contains at least one vertex from $P(v)$}~\}
    \]
    is a connected vertex cover of $G$ of size at most $k$.
    Any edge $uv$ is mapped to a distinct cycle
    \[
        u_{2p - 1} e^{1}_{uv} e^{2}_{uv} \cdots e^{8n}_{uv} v_{2q} v_{2q - 1} e^{1}_{vu} e^{2}_{vu} \cdots e^{8n}_{vu} u_{2p},
    \]
    where $v$ is the $p$th neighbor of $u$ and $u$ is the $q$th neighbor of $v$.
    The solution $S'$ needs to contain at least one vertex from this cycle.
    By the assumption that $G$ has at least two edges,
    there are at least two such cycles.
    Since $G'[S']$ is connected, at least one of the four vertices $u_{2p - 1}$, $u_{2p}$, $v_{2q - 1}$, $v_{2q}$ needs to be in $S'$.
    Thus, at least one of $u$ and $v$ is in $S$, and this verifies that $S$ is a vertex cover of $G$.

  Next, we show by contradiction that the subgraph $G[S]$ is connected.  Suppose that $X$ and $Y$ are two components of $G[S]$.  Let $Z$ be the set of neighbors of $X$, and let $Z'$ denote the set of vertices of $G'$ derived from $Z$, i.e., $Z' = \bigcup_{v\in Z} P(v)$.
  There must be a vertex $x' \in S'\cap P(x)$ for some $x\in X$ and a vertex $y' \in S'\cap P(y)$ for some $y\in Y$.
  By construction, $S'$ is disjoint from $Z'$.
  Note that $x'$ and $y'$ belong to different components of $G' - Z'$.  We have a contradiction.
    
    It remains to verify $|S| \leq k$.
    We argue that if $S'$ contains any vertex from $P(uv)$, say, $e^{i}_{uv}$ for some $i$ with $1\le i \le 8n$, then it contains all the vertices $e^{1}_{uv}$, $e^{2}_{uv}$, $\ldots$, $e^{8n}_{uv}$.
    Suppose for contradiction that $S'$ contains $e^{i}_{uv}$ but not $e^{i+1}_{uv}$ for some $i$ with $1 \leq i < 8n$.
    Since $G[S]$ is connected, $S'$ must contain $e^{1}_{uv}$ and its neighbor in $P(u)$ (again, we use the assumption that $G$ has at least two edges).
But this contradicts the minimality of $S'$ because $S' \setminus \{e^{i}_{uv}\}$ would remain a connected feedback vertex set of $G'$.
Therefore, the number of edges $e$ of $G$ such that $S'$ contains vertices from $P(e)$ is at most
    \[
      \left\lfloor \frac{k'}{8 n} \right\rfloor \le
      \left\lfloor \frac{8k + 8 n (k - 1)}{8 n} \right\rfloor = (k - 1) + \left\lfloor \frac{k}{n} \right\rfloor =  k - 1.
    \]
    Since $S'$ induces a connected subgraph, it can contain vertices from $P(v)$ for at most $k$ vertices $v$ of $G$.
This concludes the proof.
  \end{proof}

  \appendix
  \section*{Appendix: Historical remarks}\label{sec:historical-remarks}

Speckenmeyer uploaded a scanned copy of his thesis to \url{https://www.researchgate.net/publication/269029032} on March 12, 2015 (retrieved on May 26, 2023).
Before that, he had never published the gadgets in Figure~\ref{fig:speckenmeyer} or his reduction, except for a brief announcement at the end of his paper~\cite{speckenmeyer-88-fvs-and-nsis} (the reference numbers are updated to refer to the reference of the current paper):
\begin{quote}
    Garey and Johnson \cite{garey-79} have shown that the problem of determining a maximum nsis [nonseparating independent set] of a connected planar graph of degree at most 4 is $\NP$-hard and in \cite{speckenmeyer-83-thesis-fvs} the problem of determining a minimum fvs [feedback vertex set] for the same class of graphs is shown to be $\NP$-hard, too.
\end{quote}
Since Speckenmeyer's thesis was not widely available, this brief statement did not resolve the matter in the literature.
Indeed, this statement did not clarify the underlying construction.
The paper~\cite{speckenmeyer-88-fvs-and-nsis} was devoted to the correlation between feedback vertex sets and connected vertex covers.
With no access to the original source \cite{speckenmeyer-83-thesis-fvs}, it is natural to suppose that the claimed reduction is from the connected vertex cover problem on planar graphs of maximum degree four, which is not the case.
Due to the lack of an explicit construction, alternative results were published, even after three decades since Speckenmeyer's thesis.
In 2009, Rizzi \cite{rizzi-09-weakly-fundamental-cycle-bases} devised two gadgets for replacing degree-six vertices and degree-five vertices, respectively. 
Yet another reduction, directly from the satisfiability problem, was reported by Brandst{\"a}dt et al.~\cite[Theorem 1]{brandstadt-13-fvs-perfect}.
Both reductions of \cite{rizzi-09-weakly-fundamental-cycle-bases, brandstadt-13-fvs-perfect} differ from Speckenmeyer’s in that they do not preserve planarity.
The complexity question continued to generate discussion, evidenced by the several posts on \url{cstheory.stackexchange.com}, among others, asking for the complexity of the feedback vertex set problem on undirected graphs of maximum degree four.

After settling the complexity for graphs of maximum degree at least four, Speckenmeyer \cite{speckenmeyer-83-bound-cubic-fvs,speckenmeyer-88-fvs-and-nsis} then turned to graphs of maximum degree three.
Although his efforts did not directly resolve the problem, his insightful observations on correlating feedback vertex sets and connected vertex covers directly led to the subsequent resolution.
Note that the complement of a connected vertex cover is an independent set whose deletion does not disconnect the graph.
The main result of \cite{speckenmeyer-88-fvs-and-nsis} is the following
equality on a connected cubic graph $G$:
  \begin{equation}
  \label{eq:1}
    \mathtt{fvs}(G) = \mathtt{cvc}(G) - \frac{n}{2} + 1,
  \end{equation}
where $\mathtt{fvs}(G)$ and $\mathtt{cvc}(G)$ are, respectively, the cardinality of minimum feedback vertex sets of~$G$ and the cardinality of minimum connected vertex covers of $G$.
This connection motivated him to pose an open problem of determining the complexity of the feedback vertex set problem (or, equivalently, the connected vertex cover problem) on undirected cubic graphs.
Ueno, Kajitani, and Gotoh~\cite{ueno-88-cubic-fvs} solved the connected vertex cover problem by reducing it to the well-known matroid parity problem, which can be solved in polynomial time.\footnote{For the reader's convenience, two remarks on \cite{ueno-88-cubic-fvs} are in order.
  It \cite[Theorem~4]{ueno-88-cubic-fvs} mistakenly referred to~\cite{speckenmeyer-83-bound-cubic-fvs} for the proof of~\eqref{eq:1}.
  In \cite[Paragraph 3]{ueno-88-cubic-fvs}, the complement of a nonseparating independent set was mistakenly stated as a connected dominating set.  
}
In the same year, a similar algorithm was reported by Furst, Gross, and McGeoch \cite{furst88}, who studied the maximum genus of a graph, an essential invariant of topological graph theory.
It turns out that the maximum genus and minimum feedback vertex sets of a connected cubic graph are related.
Both algorithms are involved, and as far as we can check, there is no new algorithm for the feedback vertex set problem on cubic graphs.  A direct and simple algorithm would be interesting.

The final remark is on Speckenmeyer's gadgets.
As mentioned, we made a small change when rendering Figure~\ref{fig:speckenmeyer}(a).  The original one had a vertex $v_3$ similar to Figure~\ref{fig:speckenmeyer}(b), and the vertex $v_2$ receives only one instead of two edges accordingly.
Apart from being more symmetric, our rendition is easier to adapt in later constructions.  One may contrast it with Figure~\ref{fig:gadget-planar-dfvs}(b).
Further, the gadget in Figure~\ref{fig:speckenmeyer}(a) can be revised to any vertex $v$ with $d(v)\ge 6$: we assign $\lfloor {d(v) - 2 \over 2} \rfloor$ consecutive neighbors to $v_2$ and $\lceil {d(v) - 2 \over 2} \rceil$ consecutive neighbors to $v_4$.
Note the maximum degree of the gadget is
\[
  \left\lceil {d(v) - 2 \over 2} \right\rceil + 3 \le {d(v) - 1 \over 2} + 3 \le {d(v) + 5 \over 2} < d(v).
\]
Repetitively applying these reductions will eventually bring the maximum degree to five, and then the gadget in Figure~\ref{fig:speckenmeyer}(b) will complete the reduction.
Therefore, the result can be achieved without using \cite{garey-77-rectilinear-steiner-tree}.
Chen, Kanj, and Xia~\cite{chen-09-complexity-of-special-instances} used a similar idea to prove the hardness of the vertex cover on graphs of maximum degree three.

\bibliographystyle{plainurl}

\end{document}